\documentclass[11pt,reqno]{amsart}

\usepackage{amssymb,amsmath}
\usepackage[T1]{fontenc}
\usepackage[utf8]{inputenc}
\usepackage[english]{babel}

\usepackage{enumitem}
\usepackage{mathrsfs}
\usepackage{array}
\usepackage{tabularx}
\usepackage{booktabs}

\usepackage[hidelinks]{hyperref}

\newcommand{\doi}[1]{\href{https://doi.org/#1}{DOI:~#1}}

\newtheorem{theorem}{Theorem}[section]
\newtheorem{corollary}[theorem]{Corollary}
\newtheorem{lemma}[theorem]{Lemma}
\newtheorem{proposition}[theorem]{Proposition}

\theoremstyle{definition}
\newtheorem{definition}[theorem]{Definition}
\newtheorem{remark}[theorem]{Remark}
\newtheorem{example}[theorem]{Example}

\numberwithin{equation}{section}
\numberwithin{table}{section}
\numberwithin{figure}{section}

\begin{document}

%-----------------------------------------------------------
% Title
%-----------------------------------------------------------

\title[Context-Free Fixed Points and Complete Classification]
{Context-Free Fixed Points and Complete Classification of Orbits in Picard Iteration for Guarded Power Language Operators}

%-----------------------------------------------------------
% Authors
%-----------------------------------------------------------

\author{Atanas Ilchev}
\address{
Department of Mathematical Analysis,
Faculty of Mathematics and Informatics,
Paisii Hilendarski University of Plovdiv,
4000 Plovdiv, Bulgaria}
\email{atanasilchev@uni-plovdiv.bg}

\author{Hristo Kiskinov}
\address{
Department of Mathematical Analysis,
Faculty of Mathematics and Informatics,
Paisii Hilendarski University of Plovdiv,
4000 Plovdiv, Bulgaria}
\email{kiskinov@uni-plovdiv.bg}

\author{George Pashev}
\address{
Department of Computer Informatics,
Faculty of Mathematics and Informatics,
Paisii Hilendarski University of Plovdiv,
4000 Plovdiv, Bulgaria}
\email{georgepashev@uni-plovdiv.bg}

\author{Boyan Zlatanov}
\address{
Department of Mathematical Analysis,
Faculty of Mathematics and Informatics,
Paisii Hilendarski University of Plovdiv,
4000 Plovdiv, Bulgaria}
\email{bobbyz@uni-plovdiv.bg}

\thanks{Corresponding author: Atanas Ilchev.}

\date{}

%-----------------------------------------------------------
% Abstract
%-----------------------------------------------------------

\begin{abstract}
We study the language-theoretic structure of fixed points and finite
Picard iterates for guarded \(q\)-power language operators. For the
general operator, we prove that a context-free 
seeded language % seed  %%%
determines a unique
context-free fixed point and give an effective construction of a
context-free grammar generating this fixed point, independently of the
initial language.We then consider a marked single-guard special case in which two 
new % fresh
symbols separate the recursive contribution from the 
the contribution of the %NEW
seed. 
In this setting, the initial language can be traced and recovered exactly from
every finite Picard iterate by means of a regular slice and fixed-word
quotients. 
This yields injectivity of the finite-time maps and an abstract finite-time class-preservation principle. 
As a consequence, we
obtain a classification of the finite Picard iterates according to the
exact position of the initial language in the Chomsky hierarchy. Thus
the exact language-theoretic complexity may persist at every finite
stage, while all Picard orbits converge to the same context-free fixed
point.
\end{abstract}

\date{}



% Submitted and accepted dates will be inserted when available.
% \thanks{Submitted Month Day, 2026. Accepted Month Day, 2026.}

\maketitle

%===========================================================
% Introduction
%===========================================================

\section{Introduction}
\label{sec:introduction}
Formal languages, grammars, and automata admit several natural fixed-point interpretations. In particular, context-free grammars can be represented by finite systems of language equations, with the languages generated by the nonterminals forming the least solution of the corresponding system. Fixed-point principles in formal language theory have been studied for decades, both from the general viewpoint of language equations and through order-theoretic and metric methods~\cite{MollArbibKfoury1988,Kupka1997,Hesselink2010}. More recent investigations have substantially extended the theory of language equations, including equations involving concatenation and Boolean operations, their expressive power, uniqueness of solutions, and associated decision and computability problems~\cite{Okhotin2015,Okhotin2025}.

A complementary viewpoint is obtained by equipping the space of formal languages with a metric determined by the shortest word on which two languages disagree. This topology goes back to the classical work on metric spaces of languages and has subsequently been used in the study of topological properties of language classes~\cite{Vianu1977,FulopKephart2015}. Within this framework, guarded language operators were recently introduced as contractive transformations whose recursive part delays observable differences by a prescribed word length~\cite{HristovIlchevKulinaZlatanov2026}. The framework has since developed in several directions. Positive--negative guarded systems have been studied by combining the same ultrametric geometry with a suitable product order and weak monotonicity, leading to ordered fixed-point results for systems containing both positive and complemented recursive dependencies~\cite{AjetiHristovIlchevZlatanov2026}. Finite-observation conditional guarded operators provide another extension, in which the operator may depend on finitely observable properties of the input language and fixed-point behaviour can be governed by Kannan-type conditions beyond ordinary Banach contractivity~\cite{AlidemaIlchevNedelchevaZlatanov2026}. 
A further fixed-power extension yields guarded operators with a prescribed positive language power together with explicit contraction estimates, uniqueness of the fixed point, and quantitative Picard convergence~\cite{IlchevKiskinovPashevZlatanov2026}.

The present paper develops two complementary aspects of the dynamics of
guarded \(q\)-power language operators. 
We first study the 
properties of the fixed point % limiting object 
for the general operator. 
For a context-free seed, we prove that
the unique fixed point is itself context-free and 
also %NEW
give an explicit
effective construction of a context-free grammar generating it. 
The construction is independent of the initial language and shows directly
that the context-free nature of the fixed point is a structural
consequence of the fixed-point equation, rather than merely a
consequence of invariance of the context-free class.

We then turn from the 
% limiting 
fixed point to the finite Picard iterates. 
A complete classification cannot be expected for the general
operator without additional separation assumptions, since the seed 
language %NEW
may dominate the recursive contribution and destroy all information about
the initial language after finitely many steps. 
We therefore consider a
marked single-guard special case of the general \(q\)-power operator,
in which two 
new % fresh 
symbols 
for each guard %NEW
distinguish the recursive part from the seed. 
For this dynamics, we prove that the initial language can be
isolated by a regular slice and recovered exactly from every finite Picard iterate. 
This leads to injectivity of the finite-time maps, a
general class-preservation principle, and a classification of the
finite iterates according to the exact position of the initial language
in the Chomsky hierarchy. In particular, the exact language-theoretic
class of the initial state may persist at every finite stage even though
all Picard orbits converge to the same context-free fixed point.

%NEW
The paper is organized as follows. 
For the reader's convenience, Section \ref{sec:preliminaries} includes an overview of some 
standard definitions, basic properties and previously obtained results, used later in the paper.
In Section \ref{sec:context-free-fixed-points} we study the language-theoretic structure of the unique
fixed point of the guarded \(q\)-power language operator introduced in \cite{IlchevKiskinovPashevZlatanov2026}.
We prove that if the seeded language is context-free, then independent from the initial language,
the unique fixed point of this operator is context-free.
In Section \ref{sec:classification-picard-orbits} as our second main result 
the relation between the finite Picard orbits is studied and
a full  language classification of the finite Picard iterates and their fixed points
for the marked single guarded  special case of the guarded \(q\)-power language operator is obtained.
In Section \ref{sec:examples} some examples, illustrating the obtained theoretical results are presented. 
%new

%===========================================================
% Preliminaries
%===========================================================

\section{Preliminaries}
\label{sec:preliminaries}

This section collects the standard definitions, previously established results, and closure properties used in the sequel. We recall the basic language-theoretic notation, the length-based ultrametric on the space of formal languages, the previously established contraction theorem for guarded \(q\)-power operators, the Chomsky hierarchy and the relevant closure properties of its language classes, and the classical connection between context-free grammars and language equations. No substantive result in this section is claimed as new.

Throughout the paper, \(\Sigma\) denotes a fixed finite nonempty alphabet, \(\Sigma^\ast\) denotes the set of all finite words over \(\Sigma\), including the empty word \(\varepsilon\), and \(\mathcal{L}=\mathcal{P}(\Sigma^\ast)\) denotes the set of all formal languages over \(\Sigma\). For \(w\in\Sigma^\ast\), its length is denoted by \(|w|\).

\begin{definition}[{\cite{Salomaa1973,Eilenberg1974}}]
\label{def:prelim-formal-language}
A formal language over \(\Sigma\) is any subset \(L\subseteq\Sigma^\ast\). Its characteristic function is denoted by \(\mathbf{1}_L:\Sigma^\ast\to\{0,1\}\), where \(\mathbf{1}_L(w)=1\) if \(w\in L\) and \(\mathbf{1}_L(w)=0\) otherwise.
\end{definition}

For languages \(A,B\subseteq\Sigma^\ast\), their concatenation is \(AB=\{xy:x\in A,\ y\in B\}\). For \(q\ge1\), \(A^q\) denotes the \(q\)-fold concatenation of \(A\) with itself. If \(u,v\in\Sigma^\ast\) and \(L\subseteq\Sigma^\ast\), we write \(uLv=\{uwv:w\in L\}\). For a fixed word \(u\in\Sigma^\ast\), the left quotient of \(L\) by \(u\) is \(u^{-1}L=\{w\in\Sigma^\ast:uw\in L\}\), and for \(v\in\Sigma^\ast\), the right quotient is \(Lv^{-1}=\{w\in\Sigma^\ast:wv\in L\}\). These conventions are standard in formal language theory~\cite{Salomaa1973,Eilenberg1974,HopcroftMotwaniUllman2007}.

\begin{definition}[{\cite{Salomaa1973,Eilenberg1974,HopcroftMotwaniUllman2007}}]
\label{def:prelim-kleene-star}
For a language \(L\subseteq\Sigma^\ast\), its Kleene star is the language \(L^\ast=\bigcup_{n=0}^{\infty}L^n\), where \(L^0=\{\varepsilon\}\). Thus \(L^\ast\) consists of all words obtained by concatenating zero or more words from \(L\), with the concatenation of zero words understood to be \(\varepsilon\). In particular, \(\varepsilon\in L^\ast\) for every language \(L\), and \(\varnothing^\ast=\{\varepsilon\}\).
\end{definition}

% moved below!

%\begin{definition}[{\cite{Schikhof1984}}]
%\label{def:prelim-ultrametric}
%Let \(X\) be a nonempty set. A function \(\rho:X\times X\to[0,\infty)\) is an ultrametric if \(\rho(x,y)=0\) if and only if \(x=y\), \(\rho(x,y)=\rho(y,x)\), and \(\rho(x,z)\le\max\{\rho(x,y),\rho(y,z)\}\) for all \(x,y,z\in X\).
%\end{definition}

The metric used throughout the paper is the classical first-disagreement metric on formal languages considered in the Bodnarchuk--Vianu framework and in later 
works % work 
on the topology of language classes~\cite{Vianu1977,FulopKephart2015}.

\begin{definition}[{\cite{Vianu1977,FulopKephart2015}}]
\label{def:prelim-language-ultrametric}
For \(L,M\in\mathcal{L}\), put \(\nu(L,L)=\infty\), while for \(L\neq M\) define
\(
\nu(L,M)=\min\{|w|:\mathbf{1}_L(w)\neq\mathbf{1}_M(w)\}.
\)
With the convention \(2^{-\infty}=0\), define
\(
d(L,M)=2^{-\nu(L,M)}.
\)
\end{definition}

Thus \(\nu(L,M)\) is the length of 
the % a 
shortest word distinguishing \(L\) from \(M\). 
The distance therefore compares languages according to the shortest finite observation on which they disagree~\cite{Vianu1977,FulopKephart2015}.

\begin{lemma}[{\cite{Vianu1977,FulopKephart2015}}]
\label{lem:prelim-finite-depth}
Let \(L,M\in\mathcal{L}\) and \(N\ge0\). Then \(d(L,M)\le2^{-N}\) if and only if \(L\) and \(M\) agree on every word \(w\in\Sigma^\ast\) with \(|w|<N\).
\end{lemma}

\begin{definition}[{\cite{Schikhof1984}}] %moved from above
\label{def:prelim-ultrametric}
Let \(X\) be a nonempty set. A function \(\rho:X\times X\to[0,\infty)\) is an ultrametric if \(\rho(x,y)=0\) if and only if \(x=y\), \(\rho(x,y)=\rho(y,x)\), and \(\rho(x,z)\le\max\{\rho(x,y),\rho(y,z)\}\) for all \(x,y,z\in X\).
\end{definition}

\begin{theorem}[{\cite{Vianu1977,FulopKephart2015}}]
\label{thm:prelim-language-space}
The function \(d\) is an ultrametric on \(\mathcal{L}\), and the ultrametric space \((\mathcal{L},d)\) is complete. The induced topology is the Cantor-type topology obtained by identifying each language with its characteristic function in \(\{0,1\}^{\Sigma^\ast}\).
\end{theorem}

\begin{theorem}[{\cite{FulopKephart2015}}]
\label{thm:prelim-finite-dense}
The family of 
the %NEW
finite languages is dense in \((\mathcal{L},d)\).
\end{theorem}

We shall use the classical Banach contraction principle in its standard form.

\begin{theorem}[{\cite{Banach1922}}]
\label{thm:prelim-banach}
Let \((X,\rho)\) be a nonempty complete metric space and let \(T:X\to X\) satisfy \(\rho(Tx,Ty)\le c\rho(x,y)\) for all \(x,y\in X\), where \(0\le c<1\). Then \(T\) has a unique fixed point \(x^{(\star)}\). For every initial point \(x_0\in X\), the Picard iteration \(x_{n+1}=Tx_n\) converges to \(x^{(\star)}\), and
\(
\rho(x_n,x^{(\star)})\le c^n\rho(x_0,x^{(\star)})
\)
for every \(n\ge0\).
\end{theorem}

The 
% original 
guarded language operators and their contraction theory were introduced in~\cite{HristovIlchevKulinaZlatanov2026}. 
The fixed-power 
generalization % extension 
considered % required 
in the present paper was established in~\cite{IlchevKiskinovPashevZlatanov2026}.

\begin{definition}[{\cite{IlchevKiskinovPashevZlatanov2026}}]
\label{def:prelim-q-power}
Let \(q\ge1\) be an integer, let \(S\subseteq\Sigma^\ast\) be a seed language, and let
\(
G=\{(u_r,v_r):r=1,\ldots,p\}
\subseteq\Sigma^\ast\times\Sigma^\ast
\)
be a finite nonempty family of guards. 
The corresponding guarded \(q\)-power language operator \(T_q:\mathcal{L}\to\mathcal{L}\) is defined by
\[
T_q(L)
=
S\cup
\left(
\bigcup_{r=1}^{p}u_rLv_r
\right)^q.
\]
Its minimum guard length is
\(
m=\min_{1\le r\le p}(|u_r|+|v_r|).
\)
\end{definition}

For \(q=1\), Definition~\ref{def:prelim-q-power} reduces to the 
% original 
guarded language operator studied in~\cite{HristovIlchevKulinaZlatanov2026}.

\begin{theorem}[{\cite{IlchevKiskinovPashevZlatanov2026}}]
\label{thm:prelim-q-power-contraction}
Let \(T_q\) be the operator from Definition~\ref{def:prelim-q-power}, and assume that \(m\ge1\). Then
\(
d(T_q(L),T_q(M))
\le
2^{-qm}d(L,M)
\)
for all \(L,M\in\mathcal{L}\). Consequently, \(T_q\) has a unique fixed point \(L_q^{(\star)}\in\mathcal{L}\). For every initial language \(L^{(0)}\in\mathcal{L}\), the Picard iteration \(L^{(n+1)}=T_q(L^{(n)})\) converges to \(L_q^{(\star)}\), and
\(
d(L^{(n)},L_q^{(\star)})
\le
2^{-qmn}d(L^{(0)},L_q^{(\star)})
\)
for every \(n\ge0\).
\end{theorem}

\begin{remark}
\label{rem:prelim-fixed-point-notation}
To avoid confusion, we denote the fixed point of \(T_q\) by \(L_q^{(\star)}\) throughout the paper, 
whereas \(L^\ast\) denotes the Kleene star of a language \(L\).
\end{remark}

We next recall the standard 
formal  %NEW
language classes of the Chomsky hierarchy
and the closure properties that will be used 
later. %in the sequel.

\begin{definition}[{\cite{Eilenberg1974,HopcroftMotwaniUllman2007}}]
\label{def:prelim-regular-language}
A deterministic finite automaton over \(\Sigma\) is a tuple \(\mathcal{A}=(Q,\Sigma,\delta,q_0,F)\), where \(Q\) is a finite nonempty set of states, \(\delta:Q\times\Sigma\to Q\) is the transition function, \(q_0\in Q\) is the initial state, and \(F\subseteq Q\) is the set of accepting states. A language is regular if it is recognized by some deterministic finite automaton.
\end{definition}

\begin{definition}[{\cite{Salomaa1973,HopcroftMotwaniUllman2007}}]
\label{def:prelim-chomsky-hierarchy}

A formal grammar is a tuple
\(
\mathcal{G}=(V,\Sigma,P,S_0),
\)
where \(V\) is a finite set of nonterminal symbols,
\(V\cap\Sigma=\varnothing\), \(S_0\in V\) is the start symbol, and
\(P\) is a finite set of productions
\(
\alpha\to\beta,
\)
where
\(
\alpha\in(V\cup\Sigma)^\ast V(V\cup\Sigma)^\ast\)
and
\(\beta\in(V\cup\Sigma)^\ast.\)

The classes of the Chomsky hierarchy are defined as follows.

\begin{enumerate}[label=\textup{(\roman*)}]

\item
A grammar is \emph{unrestricted} (Type~\(0\)) if no further restriction
is imposed on its productions. The languages generated by such
grammars form the class
\(
\mathcal{L}_{\mathrm{un}}.
\)
Equivalently, \(\mathcal{L}_{\mathrm{un}}\) is the class of recursively
enumerable languages.

\item
A grammar is \emph{context-sensitive} (Type~\(1\)) if every production
\(\alpha\to\beta\) satisfies
\(
|\alpha|\le|\beta|,
\)
with the standard possible exception \(S_0\to\varepsilon\), provided
that \(S_0\) does not occur on the right-hand side of any production.
The corresponding class of languages is denoted by
\(
\mathcal{L}_{\mathrm{cs}}.
\)

\item
A grammar is \emph{context-free} (Type~\(2\)) if every production has
the form
\[
A\to\alpha,
\qquad
A\in V,\quad
\alpha\in(V\cup\Sigma)^\ast.
\]
The corresponding class of languages is denoted by
\(
\mathcal{L}_{\mathrm{cf}}.
\)

\item
A grammar is \emph{regular} (Type~\(3\)) if it is right-linear; its
productions may be taken in the form
\(A\to aB\) or \(A\to a\), where \(A,B\in V\) and \(a\in\Sigma\),
with the standard possible exception \(S_0\to\varepsilon\), provided
that \(S_0\) does not occur on the right-hand side of any production.
The corresponding class of languages is denoted by
\(
\mathcal{L}_{\mathrm{reg}}.
\)

\end{enumerate}

Thus the Chomsky hierarchy gives the inclusions
$$
\mathcal{L}_{\mathrm{reg}}
\subset
\mathcal{L}_{\mathrm{cf}}
\subset
\mathcal{L}_{\mathrm{cs}}
\subset
\mathcal{L}_{\mathrm{un}}
\subset
\mathcal{L}.
$$

It is well known that 
for $|\Sigma|>1$ (as in the case considered by us later)  %NEW
all these inclusions are strict.

For later use, we also denote by
\(
\mathcal{L}_{\mathrm{nrn}}
=
\mathcal{L}\setminus\mathcal{L}_{\mathrm{un}}
\)
the class of languages that are not recursively enumerable.
\end{definition}

\begin{theorem}[{\cite{Salomaa1973,HopcroftMotwaniUllman2007}}]
\label{thm:prelim-chomsky-closure}
Each of the classes
\(\mathcal{L}_{\mathrm{reg}},\) \ 
\(\mathcal{L}_{\mathrm{cf}},\) \ 
\(\mathcal{L}_{\mathrm{cs}},\) \ 
\(\mathcal{L}_{\mathrm{un}}\)
is closed under finite union and concatenation, under intersection
with a regular language, and under left and right quotient by a fixed
word.
\end{theorem}

\begin{theorem}[{\cite{Salomaa1973,Eilenberg1974,HopcroftMotwaniUllman2007}}]
\label{thm:prelim-cfl-closure}
The class of context-free languages is closed under finite union, concatenation, Kleene star, homomorphism, inverse homomorphism, intersection with a regular language, and left and right quotient by a regular language. 
In particular, if \(L\in\mathcal{L}_{\mathrm{cf}}\) and \(u,v\in\Sigma^\ast\), then \(uLv\) is context-free, and the left and right quotients of \(L\) by fixed words are context-free.
\end{theorem}

As an immediate consequence of the standard closure properties in Theorem~\ref{thm:prelim-cfl-closure}, the guarded \(q\)-power operator from Definition~\ref{def:prelim-q-power} 
preserves context-freeness whenever its seed is context-free.

\begin{corollary}%[{\cite{Salomaa1973,Eilenberg1974,HopcroftMotwaniUllman2007,IlchevKiskinovPashevZlatanov2026}}]
\label{cor:prelim-cfl-invariance}
If \(S\in\mathcal{L}_{\mathrm{cf}}\), then
\(
T_q(\mathcal{L}_{\mathrm{cf}})
\subseteq
\mathcal{L}_{\mathrm{cf}}.
\)
\end{corollary}

\begin{lemma}[{\cite{Salomaa1973,HopcroftMotwaniUllman2007}}]
\label{lem:prelim-anbncn}
If the alphabet contains three distinct symbols \(a,b,c\), then the language
\(
\{a^n b^n c^n:n\ge1\}
\)
is not context-free.
\end{lemma}

Finite languages are regular and therefore context-free~\cite{Eilenberg1974,HopcroftMotwaniUllman2007}. 
Hence Theorem~\ref{thm:prelim-finite-dense} implies that \(\mathcal{L}_{\mathrm{cf}}\) is dense in \(\mathcal{L}\). 
It is also standard that over every finite nonempty alphabet there exist languages that are not context-free~\cite{Salomaa1973,HopcroftMotwaniUllman2007}.

\begin{corollary}
\label{cor:prelim-cfl-incomplete}
The metric subspace \((\mathcal{L}_{\mathrm{cf}},d)\) is incomplete.
\end{corollary}

\begin{remark}
\label{rem:prelim-cfl-incompleteness}
The preceding corollary is a consequence of the standard density and language-class properties recalled above~\cite{FulopKephart2015,Salomaa1973,HopcroftMotwaniUllman2007}. Its relevance here is that invariance of the context-free class under an operator does not, by itself, ensure that the limit of a convergent Picard orbit is context-free. The context-free nature of the limiting fixed point therefore requires a separate justification.
\end{remark}

The 
following %NEW
algebraic interpretation of context-free grammars by systems of language equations is 
well known. % classical.

\begin{theorem}[{\cite{Hesselink2010,Okhotin2015}}]
\label{thm:prelim-language-equations}
Every context-free grammar can be represented by a finite system of language equations whose right-hand sides are formed from finite unions and concatenations of variables and terminal words. 
The languages generated by the nonterminals of the grammar form the least solution of the corresponding system. Conversely, finite systems of this grammar form provide an equivalent algebraic representation of context-free grammars.
\end{theorem}

%===========================================================
% Context-free fixed points
%===========================================================
\section{Context-Free Fixed Points of Guarded \(q\)-Power Language Operators}
\label{sec:context-free-fixed-points}

In this section we study the language-theoretic structure of the unique
fixed point of the general guarded \(q\)-power language operator from
Definition~\ref{def:prelim-q-power}. 
Throughout the section, \(T_q\)
denotes this operator and its minimum guard length is assumed to satisfy
\(m\ge1\). 
By Theorem~\ref{thm:prelim-q-power-contraction}, 
there exists %NEW
a unique % the 
fixed point
\(L_q^{(\star)}\in\mathcal{L}\) 
% exists uniquely,
 and every Picard orbit
converges to it independently of the chosen initial language.

The contraction theorem determines existence, uniqueness, and
convergence, but it does not determine the language-theoretic nature
of the fixed point. In particular, although
Corollary~\ref{cor:prelim-cfl-invariance} shows that the context-free
class is invariant under \(T_q\), this alone is insufficient to conclude
that the limiting language is context-free, since the context-free
subspace is incomplete in the inherited ultrametric. 
Therefore we % We therefore
begin
by determining the language class of the unique fixed point directly.

\begin{theorem}
\label{thm:cfl-fixed-point}
Let \(q\ge1\) and \(m\ge1\). 
If the seed language \(S\) is context-free, then the unique fixed point \(L_q^{(\star)}\) of \(T_q\) is context-free.
\end{theorem}

\begin{proof}
Let \(\mathcal{G}_S=(V,\Sigma,P,S_0)\) be a context-free grammar generating \(S\). 
Choose two 
new % fresh 
symbols \(X,B\notin V\cup\Sigma\), regard them as new nonterminals, 
and define the grammar \(\mathcal{G}_q'=(V\cup\{X,B\},\Sigma,P_q',X)\), where
\[
P_q'=P\cup\{X\to S_0,\;X\to B^q\}\cup\{B\to u_rXv_r:r=1,\ldots,p\}.
\]
Here \(B^q\) denotes the word consisting of \(q\) consecutive occurrences of the nonterminal \(B\). Since \(V\) and \(P\) are finite and only finitely many new nonterminals and productions are added, \(\mathcal{G}_q'\) is a context-free grammar.

Let \(K\) denote the language generated from \(X\). 
Since the only productions with left-hand side \(B\) are \(B\to u_rXv_r\), 
the language generated from \(B\) is exactly \(\bigcup_{r=1}^{p}u_rKv_r\). 
Hence \(X\to B^q\) generates exactly \(\left(\bigcup_{r=1}^{p}u_rKv_r\right)^q\), 
while \(X\to S_0\) generates \(S\). Therefore
\[
K=S\cup\left(\bigcup_{r=1}^{p}u_rKv_r\right)^q=T_q(K).
\]
Thus \(K\) is a fixed point of \(T_q\).

According % By 
Theorem~\ref{thm:prelim-q-power-contraction}, \(T_q\) has a unique fixed point 
$L_q^{(\star)}$ %NEW
in \(\mathcal{L}\). 
Consequently, \(K=L_q^{(\star)}\). Since \(\mathcal{G}_q'\) is a context-free grammar generating \(K\), 
the language \(L_q^{(\star)}\) is context-free.
\end{proof}

\begin{proposition}
\label{prop:effective-fixed-point-grammar}
Let \(q\ge1\), \(m\ge1\), and suppose that \(S\) is generated by a context-free grammar \(\mathcal{G}_S=(V,\Sigma,P,S_0)\). A context-free grammar for \(L_q^{(\star)}\) can be constructed effectively from \(\mathcal{G}_S\), \(q\), and \(G\). The construction introduces exactly two additional nonterminals and \(p+2\) additional productions. In particular, it produces a grammar with \(|V|+2\) nonterminals and \(|P|+p+2\) productions.
\end{proposition}

\begin{proof}
Use the grammar \(\mathcal{G}_q'\) constructed in the proof of Theorem~\ref{thm:cfl-fixed-point}. Besides the productions of \(P\), it contains the two productions \(X\to S_0\) and \(X\to B^q\), together with one production \(B\to u_rXv_r\) for each \(r=1,\ldots,p\). Thus exactly two new nonterminals and \(p+2\) new productions are introduced. Theorem~\ref{thm:cfl-fixed-point} shows that this grammar generates \(L_q^{(\star)}\).
\end{proof}

\begin{remark}
\label{rem:effective-independence}
The construction in Proposition~\ref{prop:effective-fixed-point-grammar} depends only on a grammar for the seed, the fixed exponent \(q\), and the guard family. It is completely independent of the initial language used in the Picard iteration. In particular, no effective representation of the initial language is required in order to construct a context-free grammar generating the limiting fixed point.
\end{remark}

\begin{corollary}
\label{cor:arbitrary-initial-language}
Let \(q\ge1\), \(m\ge1\), and assume that \(S\) is context-free. For every initial language \(L^{(0)}\in\mathcal{L}\), 
the Picard iteration \(L^{(n+1)}=T_q(L^{(n)})\) converges to the same context-free fixed point \(L_q^{(\star)}\), and
\begin{equation}
d(L^{(n)},L_q^{(\star)})\le2^{-qmn}d(L^{(0)},L_q^{(\star)})
\qquad\text{for every }n\ge0.
\label{eq:picard-cfl-limit}
\end{equation}
\end{corollary}

\begin{proof}
The convergence and estimate~\eqref{eq:picard-cfl-limit} follow from Theorem~\ref{thm:prelim-q-power-contraction}, while Theorem~\ref{thm:cfl-fixed-point} shows that the common limit \(L_q^{(\star)}\) is context-free.
\end{proof}

\begin{remark}
\label{rem:arbitrary-complexity}
No language-theoretic restriction is imposed on \(L^{(0)}\) in Corollary~\ref{cor:arbitrary-initial-language}. 
In particular, the initial language need not be context-free, decidable, or recursively enumerable. 
It may be an arbitrary element of \(\mathcal{P}(\Sigma^\ast)\), 
while the unique limit of its Picard orbit is nevertheless context-free whenever the seed is context-free. 
Thus the language-theoretic complexity of the initial state may affect the finite orbit 
but does not determine the language class of its unique limit.
\end{remark}

\begin{corollary}
\label{cor:cfl-picard-orbit}
Let \(q\ge1\), \(m\ge1\), and assume that both \(S\) and \(L^{(0)}\) are context-free. Then every Picard iterate \(L^{(n)}\) is context-free and \(L^{(n)}\to L_q^{(\star)}\), where \(L_q^{(\star)}\) is context-free.
\end{corollary}

\begin{proof}
The assertion follows by induction from Corollary~\ref{cor:prelim-cfl-invariance}. 
The initial language \(L^{(0)}\) is context-free by assumption, and whenever \(L^{(n)}\) is context-free, Corollary~\ref{cor:prelim-cfl-invariance} gives \(L^{(n+1)}=T_q(L^{(n)})\in\mathcal{L}_{\mathrm{cf}}\). 
The convergence % Convergence 
follows from Corollary~\ref{cor:arbitrary-initial-language}
 and Theorem~\ref{thm:cfl-fixed-point} gives \(L_q^{(\star)}\in\mathcal{L}_{\mathrm{cf}}\).
\end{proof}

\begin{remark}
\label{rem:not-completeness}
The conclusion \(L_q^{(\star)}\in\mathcal{L}_{\mathrm{cf}}\) does not follow merely 
from Corollary~\ref{cor:cfl-picard-orbit}. 
A convergent sequence of context-free languages may have 
a non-context-free limit in the complete ambient language space because \((\mathcal{L}_{\mathrm{cf}},d)\) is incomplete. Theorem~\ref{thm:cfl-fixed-point} is therefore a structural statement about the fixed-point equation itself rather than a consequence of invariance of the context-free class.
\end{remark}

%===========================================================
% Classification of finite Picard iterates
%===========================================================
\section{Classification of Language Classes of the Finite Picard Iterates
and of the Fixed Point 
for a %: A 
Special Case of the Guarded \(q\)-Power Operator}
\label{sec:classification-picard-orbits}

The preceding section determines the language-theoretic nature of the
unique fixed point for the general guarded \(q\)-power operator.
We now turn to a different question: what can be said about the
language classes of the finite Picard iterates before the 
limit % limiting 
fixed point is reached?

For the general operator from Definition~\ref{def:prelim-q-power},
no classification of the finite iterates solely in terms of the
initial language can hold without additional assumptions. Indeed,
consider the admissible choice
\(
S=\Sigma^\ast.
\)
Then, for every language \(L\subseteq\Sigma^\ast\),
\(
T_q(L)
=
\Sigma^\ast
\cup
\left(
\bigcup_{r=1}^{p}u_rLv_r
\right)^q
=
\Sigma^\ast.
\)
Consequently, independently of the initial language \(L^{(0)}\),
one has
\(
L^{(1)}=\Sigma^\ast
\)
and all subsequent Picard iterates coincide with the fixed point
\(\Sigma^\ast\)
which is regular. %NEW
Thus the seed may completely dominate the recursive
part of the operator and erase, already after one step, any
language-theoretic distinction between different initial languages.

The preceding example shows that the general operator \(T_q\) may
erase all information about the initial language after a single
Picard step. 
To prevent the seed from masking the recursive
contribution, we introduce a marked special case, denoted by
\(T_{qm}\). 
Two 
new % fresh 
symbols 
mark the left and right guards and %NEW
distinguish the recursively generated
words from the seed, while both the seed and the initial language
remain over the original alphabet. 
This separation is the basis
for the 
possibility of %NEW
exact recovery of the initial language from 
every % NEW
finite Picard
iterate % iterates 
and for the classification established below under the
corresponding assumptions on the seed
and the initial languages. %NEW 

\begin{definition}
\label{def:marked-q-power}
Let \(\Sigma\) be a finite nonempty alphabet, let \(x,y\notin\Sigma\)
be two distinct 
new % fresh 
symbols, and put
\(
\Gamma=\Sigma\cup\{x,y\}.
\)
Let \(q\ge1\) be an integer and let \(S\subseteq\Sigma^\ast\) be a
seed language. The marked single-guard \(q\)-power language operator
\(
T_{qm}:\mathcal{P}(\Gamma^\ast)\to\mathcal{P}(\Gamma^\ast)
\)
is defined 
for $L\subseteq\Gamma^\ast$ %NEW
by
\[T_{qm}(L)=S\cup(xLy)^q.    %,\  L\subseteq\Gamma^\ast.
\]
Here the letter \(m\) in the subscript is a fixed label indicating
the marked case and does not denote the minimum guard length.
\end{definition}

The operator \(T_{qm}\) is the special case of
Definition~\ref{def:prelim-q-power} over the enlarged alphabet
\(\Gamma\) obtained by taking \(p=1\) and the single guard
\(
(u_1,v_1)=(x,y).
\)
The seed and the initial language are assumed to satisfy
\(
S,L^{(0)}\subseteq\Sigma^\ast\subset\Gamma^\ast.
\)

Throughout the remainder of this section, the ambient alphabet is
\(\Gamma\). Accordingly, we write
\(
\mathcal{L}=\mathcal{P}(\Gamma^\ast),
\)
and the symbols
\(\mathcal{L}_{\mathrm{reg}},\) \ 
\(\mathcal{L}_{\mathrm{cf}},\) \ 
\(\mathcal{L}_{\mathrm{cs}},\) \ 
\(\mathcal{L}_{\mathrm{un}}\)
refer to the corresponding language classes over \(\Gamma\).
Languages over the original alphabet \(\Sigma\) are identified with
their natural copies in \(\Gamma^\ast\). This enlargement of the
ambient alphabet does not change their membership in any of the
classes of the Chomsky hierarchy. In particular,
\(
\mathcal{L}_{\mathrm{nrn}}
=
\mathcal{L}\setminus\mathcal{L}_{\mathrm{un}}
\)
denotes, in this section, the class of non-recursively enumerable
languages over \(\Gamma\).

The same convention applies to the language ultrametric. Throughout this section, \(d\) denotes the first-disagreement ultrametric of Definition~\ref{def:prelim-language-ultrametric} on \(\mathcal{P}(\Gamma^\ast)\), with word lengths taken in \(\Gamma^\ast\). Consequently, Theorem~\ref{thm:prelim-language-space} and the contraction result in Theorem~\ref{thm:prelim-q-power-contraction} apply to the marked operator over the ambient alphabet \(\Gamma\).

The role of the 
new % fresh 
symbols \(x\) and \(y\) is to mark the recursive
part of the dynamics. Since \(S\subseteq\Sigma^\ast\), no word of the
seed contains either marker, whereas every word produced through the
recursive term \((xLy)^q\) necessarily contains both \(x\) and \(y\).
Hence
\(
S\cap(xLy)^q=\varnothing
\)
for every \(L\subseteq\Gamma^\ast\).
Thus, once information from a previous Picard iterate passes through
the guard, it remains distinguishable from the newly inserted seed
contribution.

The results below show that this separation has a strong consequence:
the initial language remains exactly traceable and recoverable from
every finite Picard iterate. This will ultimately allow us to classify
the finite orbit according to the exact position of the initial
language in the Chomsky hierarchy, while the limiting fixed point
remains context-free.

\begin{remark}
\label{rem:marked-operator-fixed-point}
Since the marked operator is 
a special case % an instance 
of the general guarded
\(q\)-power construction over the ambient alphabet \(\Gamma\), all
fixed-point results from the preceding section remain applicable.
In particular, whenever \(S\in\mathcal{L}_{\mathrm{cf}}\), the operator
has a unique context-free fixed point \(L_q^{(\star)}\), and every Picard
orbit converges to this same fixed point independently of the initial
language.
\end{remark}

For the remainder of this section, we use the following conventions.
Let \(T_{qm}\) be the marked operator from
Definition~\ref{def:marked-q-power}, with fixed integer \(q\ge1\)
and seed \(S\subseteq\Sigma^\ast\). The alphabets \(\Sigma\) and
\(\Gamma=\Sigma\cup\{x,y\}\) are as specified in that definition.
We write \(\mathcal{L}=\mathcal{P}(\Gamma^\ast)\), and all language
classes are understood over \(\Gamma\).
For an arbitrary initial language \(H\subseteq\Sigma^\ast\), its
Picard orbit is denoted by
\(
L^{(0)}=H,\quad
L^{(n+1)}=T_{qm}(L^{(n)}),\quad n\ge0.
\)
The unique fixed point of \(T_{qm}\) is denoted by
\(L_q^{(\star)}\).

For each word \(z\in\Sigma^\ast\), define the auxiliary words
\(z_n,\alpha_n,\beta_n\) by
\(z_0=z\), \(\alpha_0=\beta_0=\varepsilon\), and
\(z_{n+1}=(xz_ny)^q,\quad \alpha_{n+1}=x\alpha_n,\quad \beta_{n+1}=\beta_n y(xz_ny)^{q-1},\) \  \(n\ge0.\)
Here \(w^0=\varepsilon\) for every word \(w\).
We also set
\(
R_n=\alpha_n\Sigma^\ast\beta_n
\)
for every \(n\ge0\).

\begin{theorem}
\label{thm:finite-time-trace}
Let \(H\ne\varnothing\) and choose \(z\in H\).
With the notation above, for every \(n\ge0\), one has
\(z_n\in L^{(n)}\), the language \(R_n\) is regular, and
\begin{equation}
L^{(n)}\cap R_n=\alpha_nH\beta_n.
\label{eq:finite-time-trace}
\end{equation}
\end{theorem}

\begin{proof}
We first show that \(z_n\in L^{(n)}\) for every \(n\ge0\).
Since \(z_0=z\in H=L^{(0)}\), the assertion holds for \(n=0\).
If \(z_n\in L^{(n)}\), then
\(
z_{n+1}=(xz_ny)^q\in(xL^{(n)}y)^q\subseteq L^{(n+1)}.
\)
Hence the assertion follows by induction.

Since
\(
R_n=\alpha_n\Sigma^\ast\beta_n
\)
is obtained from the regular language \(\Sigma^\ast\) by concatenation with the two fixed words \(\alpha_n\) and \(\beta_n\), the language \(R_n\) is regular for every \(n\ge0\).

Let
\(
\pi:\Gamma^\ast\longrightarrow\{x,y\}^\ast
\)
be the homomorphism 
that erases every symbol of the original alphabet \(\Sigma\) and keeps only the symbols \(x\) and \(y\). 
For a word \(v\in\{x,y\}^\ast\), define its balance by
\(
\delta(v)=|v|_x-|v|_y.
\)
We establish two structural properties.

First, for every \(k\ge0\) and every \(w\in L^{(k)}\), the projection \(\pi(w)\) is a Dyck word: its total balance is zero and every prefix has nonnegative balance.

For \(k=0\), one has \(L^{(0)}=H\subseteq\Sigma^\ast\). Since neither \(x\) nor \(y\) occurs in any word of \(H\), one has \(\pi(w)=\varepsilon\) for every \(w\in L^{(0)}\).

Assume that the assertion holds for every word of \(L^{(k)}\). Since \(S\subseteq\Sigma^\ast\), every seed word \(s\in S\) also satisfies \(\pi(s)=\varepsilon\). Every other word of \(L^{(k+1)}\) belongs to \((xL^{(k)}y)^q\) and is therefore a concatenation of \(q\) blocks of the form \(xwy\), with \(w\in L^{(k)}\). By the induction hypothesis, \(\pi(w)\) is a Dyck word, and therefore \(x\pi(w)y\) is again a Dyck word. A concatenation of \(q\) such words is also a Dyck word. Hence the assertion holds for \(L^{(k+1)}\).

Moreover, if \(w\in L^{(k)}\), then \(x\pi(w)y\) has strictly positive balance on every proper nonempty prefix and returns to balance zero only at its final symbol. Consequently, every word in \((xL^{(k)}y)^q\) has a unique decomposition into its \(q\) top-level factors
\(
xw_1y,\ldots,xw_qy,\)
\(w_i\in L^{(k)},\)
and the factor boundaries are determined uniquely by the successive returns of the \(x\)-\(y\) balance to zero.

Second, we claim that
\(
\pi(r)=\pi(z_n)
\)
for every \(r\in R_n\).

For \(n=0\), one has \(R_0=\Sigma^\ast\), while \(z_0=z\in H\subseteq\Sigma^\ast\). Hence
\(
\pi(r)=\varepsilon=\pi(z_0)
\)
for every \(r\in R_0\).

Assume that the claim holds for some \(n\ge0\). From the recursive definitions of \(\alpha_n\) and \(\beta_n\),
\(
R_{n+1}=xR_ny(xz_ny)^{q-1}.
\)
Thus every \(r'\in R_{n+1}\) has the form
\(
r'=xry(xz_ny)^{q-1}
\)
for some \(r\in R_n\). Using the induction hypothesis gives
\[
\pi(r')
=
x\pi(r)y\bigl(x\pi(z_n)y\bigr)^{q-1}
=
\bigl(x\pi(z_n)y\bigr)^q
=
\pi(z_{n+1}).
\]
This proves the second structural property.

We now prove~\eqref{eq:finite-time-trace} by induction on \(n\).

For \(n=0\), one has \(R_0=\Sigma^\ast\) and
\(L^{(0)}=H\subseteq\Sigma^\ast\). Therefore
\(
L^{(0)}\cap R_0=H=\alpha_0H\beta_0.
\)

Assume that~\eqref{eq:finite-time-trace} holds for some \(n\ge0\).
Since
\(
R_{n+1}=xR_ny(xz_ny)^{q-1},
\)
every word of \(R_{n+1}\) contains the new symbol \(x\), whereas every word of the seed \(S\) belongs to \(\Sigma^\ast\) and therefore contains neither \(x\) nor \(y\). Consequently,
\(
R_{n+1}\cap S=\varnothing.
\)
Hence every word in \(L^{(n+1)}\cap R_{n+1}\) must belong to the recursive part \((xL^{(n)}y)^q\).

Every word in \(R_{n+1}\) has the form
\(
xry(xz_ny)^{q-1},\)
\(r\in R_n.\)
By the second structural property,
\(
\pi(r)=\pi(z_n).
\)
Since \(\pi(z_n)\) is a Dyck word, the projection of \(xry\) is primitive: its balance is strictly positive on every proper nonempty prefix and returns to zero only at the final \(y\). Therefore \(xry\) occupies exactly one top-level block in the unique decomposition described above. When \(q>1\), it is followed by exactly \(q-1\) top-level blocks equal to \(xz_ny\).

Suppose that the displayed word also belongs to \((xL^{(n)}y)^q\). By uniqueness of the top-level decomposition, its first interior word must satisfy \(r\in L^{(n)}\), while the remaining \(q-1\) interior words are \(z_n\in L^{(n)}\). Since also \(r\in R_n\), we obtain
\(
L^{(n+1)}\cap R_{n+1}
\subseteq
x\bigl(L^{(n)}\cap R_n\bigr)y(xz_ny)^{q-1}.
\)

Conversely, let \(r\in L^{(n)}\cap R_n\). Since \(z_n\in L^{(n)}\),
\(
xry(xz_ny)^{q-1}\in(xL^{(n)}y)^q.
\)
The same word belongs to
\(
R_{n+1}=xR_ny(xz_ny)^{q-1},
\)
and hence
\(
xry(xz_ny)^{q-1}\in L^{(n+1)}\cap R_{n+1}.
\)
Therefore
\(
L^{(n+1)}\cap R_{n+1}
=
x\bigl(L^{(n)}\cap R_n\bigr)y(xz_ny)^{q-1}.
\)

Applying the induction hypothesis gives
\[
L^{(n+1)}\cap R_{n+1}=x\alpha_nH\beta_n y(xz_ny)^{q-1}=\alpha_{n+1}H\beta_{n+1}.
\]
Thus~\eqref{eq:finite-time-trace} holds for every \(n\ge0\).
\end{proof}

The preceding trace identity allows us to go one step further: the initial language can be reconstructed exactly from every finite Picard iterate. Thus no information about the initial state is lost at any finite stage of the dynamics.

To compare initial languages containing a common distinguished word,
for each \(z\in\Sigma^\ast\) we set
\(
\mathfrak{H}_z=\{H\subseteq\Sigma^\ast:z\in H\}.
\)
For \(H\in\mathfrak{H}_z\), we write
\(L_H^{(0)}=H\) and
\(L_H^{(n+1)}=T_{qm}(L_H^{(n)})\)
to indicate explicitly the dependence of the Picard orbit on \(H\).
For every \(n\ge0\), define
\(
F_n:\mathfrak{H}_z\to\mathcal{P}(\Gamma^\ast)
\)
by \(F_n(H)=L_H^{(n)}\).
With \(z\) fixed, the auxiliary words \(z_n,\alpha_n,\beta_n\)
and the regular language \(R_n\) are the same for all
\(H\in\mathfrak{H}_z\).

\begin{theorem}
\label{thm:finite-time-recoverability}
Let \(z\in\Sigma^\ast\). For every \(H\in\mathfrak{H}_z\) and
\(n\ge0\), the initial language is recovered from its Picard
iterate by
\begin{equation}
H=
\left(
\alpha_n^{-1}
\bigl(L^{(n)}\cap R_n\bigr)
\right)
\beta_n^{-1}.
\label{eq:exact-recovery}
\end{equation}
Consequently, the map
\(F_n:\mathfrak{H}_z\to\mathcal{P}(\Gamma^\ast)\)
is injective for every \(n\ge0\).
\end{theorem}

\begin{proof}
Let \(H\subseteq\Sigma^\ast\) be nonempty and choose \(z\in H\).
By Theorem~\ref{thm:finite-time-trace},
\[
L^{(n)}\cap R_n
=
\alpha_nH\beta_n
\qquad
\text{for every }n\ge0.
\]
Taking the left quotient by the fixed word \(\alpha_n\) gives
\[
\alpha_n^{-1}
\bigl(L^{(n)}\cap R_n\bigr)
=
\alpha_n^{-1}
(\alpha_nH\beta_n)
=
H\beta_n.
\]
Indeed, cancellation in the free monoid implies that
\(
\alpha_n w\in\alpha_nH\beta_n
\)
if and only if
\(
w\in H\beta_n.
\)
Taking next the right quotient by the fixed word \(\beta_n\) yields
\[
\left(
\alpha_n^{-1}
\bigl(L^{(n)}\cap R_n\bigr)
\right)
\beta_n^{-1}
=
(H\beta_n)\beta_n^{-1}
=
H.
\]
This proves~\eqref{eq:exact-recovery}.

We now prove the injectivity statement. Fix \(z\in\Sigma^\ast\) and
let \(H_1,H_2\in\mathfrak{H}_z\). Since the same distinguished word
\(z\) belongs to both languages, the sequences
\(z_n,\alpha_n,\beta_n\) and the regular language
\(R_n=\alpha_n\Sigma^\ast\beta_n\) are identical for the two Picard
orbits.

Suppose that, for some finite \(n\),
\(
F_n(H_1)=F_n(H_2).
\)
Applying the recovery formula~\eqref{eq:exact-recovery} to this common
finite iterate gives
%\[
%\begin{aligned}
%H_1
%&=
%\left(
%\alpha_n^{-1}
%\bigl(F_n(H_1)\cap R_n\bigr)
%\right)
%\beta_n^{-1}
%\\
%&=
%\left(
%\alpha_n^{-1}
%\bigl(F_n(H_2)\cap R_n\bigr)
%\right)
%\beta_n^{-1}
%=
%H_2.
%\end{aligned}
%\]

$$ % rewritten from above
H_1
=
\left(
\alpha_n^{-1}
\bigl(F_n(H_1)\cap R_n\bigr)
\right)
\beta_n^{-1}
=
\left(
\alpha_n^{-1}
\bigl(F_n(H_2)\cap R_n\bigr)
\right)
\beta_n^{-1}
=
H_2.
$$

Therefore \(F_n\) is injective for every finite \(n\).
\end{proof}

\begin{corollary}
\label{cor:injective-uniform-collapse}
Under the assumptions of Theorem~\ref{thm:finite-time-recoverability},
every finite-time map
\[
F_n:\mathfrak{H}_z\to\mathcal{P}(\Gamma^\ast)
\]
is injective, while
\begin{equation}
\sup_{H\in\mathfrak{H}_z}
d\bigl(F_n(H),L_q^{(\star)}\bigr)
\le
2^{-2qn}.
\label{eq:uniform-collapse}
\end{equation}
Consequently, the sequence of injective maps \(F_n\) converges uniformly on \(\mathfrak{H}_z\) to the constant map
\(
F_\infty(H)=L_q^{(\star)}.
\)
\end{corollary}

\begin{proof}
Injectivity follows from Theorem~\ref{thm:finite-time-recoverability}. In the present construction the single guard is \((x,y)\), so the minimum guard length is \(m=2\). The contraction estimate from Theorem~\ref{thm:prelim-q-power-contraction} therefore gives
\(
d\bigl(F_n(H),L_q^{(\star)}\bigr)
\le
2^{-2qn}d(H,L_q^{(\star)})
\le
2^{-2qn},
\)
because the language ultrametric takes values in \([0,1]\). The bound is independent of \(H\), which proves~\eqref{eq:uniform-collapse} and hence uniform convergence to the constant map \(H\mapsto L_q^{(\star)}\).
\end{proof}

\begin{remark}
\label{rem:finite-information-limit-collapse}
Corollary~\ref{cor:injective-uniform-collapse} separates finite-time information from asymptotic behaviour. 
At every finite stage the Picard dynamics still distinguishes all initial languages in \(\mathfrak{H}_z\), 
since \(F_n\) is injective and the initial language can be reconstructed exactly. Nevertheless, the same maps converge uniformly to a constant map. 
Thus no finite stage loses the encoded initial information, although this information disappears completely from the 
limit % limiting 
state.
\end{remark}

We now turn to the preservation of language classes under the
marked Picard iteration. For a class
\(\mathscr{C}\subseteq\mathcal{L}\), we consider the following
recovery closure properties:
\begin{enumerate}
\item[(R1)]
If \(L\in\mathscr{C}\) and \(R\subseteq\Gamma^\ast\) is regular,
then \(L\cap R\in\mathscr{C}\).
\item[(R2)]
If \(L\in\mathscr{C}\) and \(u,v\in\Gamma^\ast\) are fixed words,
then \(u^{-1}L\in\mathscr{C}\) and \(Lv^{-1}\in\mathscr{C}\).
\end{enumerate}

\begin{theorem}\label{thm:finite-time-class-preservation}
Let \(H\ne\varnothing\), and let
\(\mathscr{C}\subseteq\mathcal{L}\) satisfy
\textup{(R1)}--\textup{(R2)}. 

Then, for every \(n\ge0\),
\begin{equation}
L^{(n)}\in\mathscr{C}
\quad\Longrightarrow\quad
H\in\mathscr{C}
\label{eq:backward-class-transfer-general}
\end{equation}
and % Equivalently, for every \(n\ge0\),
\begin{equation}
H\notin\mathscr{C}
\quad\Longrightarrow\quad
L^{(n)}\notin\mathscr{C}.
\label{eq:class-nonmembership-persistence}
\end{equation}
If, in addition,
\begin{equation}
T_{qm}(\mathscr{C})\subseteq\mathscr{C},
\label{eq:forward-class-invariance}
\end{equation}
then, for every \(n\ge0\),
\begin{equation}
H\in\mathscr{C}
\quad\Longleftrightarrow\quad
L^{(n)}\in\mathscr{C}.
\label{eq:finite-time-class-equivalence}
\end{equation}
\end{theorem}

\begin{proof}
We first prove the backward implication
\eqref{eq:backward-class-transfer-general}. Let \(n\ge0\) be fixed
and suppose that
\(
L^{(n)}\in\mathscr{C}.
\)

By Theorem~\ref{thm:finite-time-recoverability}, there exist fixed
words \(\alpha_n,\beta_n\in\Gamma^\ast\) and a regular language
\(
R_n=\alpha_n\Sigma^\ast\beta_n
\)
such that
\begin{equation}
H=
\left(
\alpha_n^{-1}
\bigl(L^{(n)}\cap R_n\bigr)
\right)
\beta_n^{-1}.
\label{eq:abstract-class-recovery}
\end{equation}

Since \(R_n\) is regular and
\(L^{(n)}\in\mathscr{C}\), property~(R1) gives
\(
L^{(n)}\cap R_n\in\mathscr{C}.
\)
Applying the left quotient by the fixed word \(\alpha_n\) and using
property~(R2), we obtain
\(
\alpha_n^{-1}
\bigl(L^{(n)}\cap R_n\bigr)
\in\mathscr{C}.
\)
A second application of~(R2), now to the right quotient by the fixed
word \(\beta_n\), gives
\(
\left(
\alpha_n^{-1}
\bigl(L^{(n)}\cap R_n\bigr)
\right)
\beta_n^{-1}
\in\mathscr{C}.
\)
By~\eqref{eq:abstract-class-recovery}, the language on the left-hand
side is exactly \(H\). Hence
\(
H\in\mathscr{C}.
\)
This proves~\eqref{eq:backward-class-transfer-general}.

Taking the contrapositive immediately yields
\(
H\notin\mathscr{C}
\Longrightarrow
L^{(n)}\notin\mathscr{C}
\)
for every finite \(n\ge0\), which proves
\eqref{eq:class-nonmembership-persistence}.

Assume now, in addition, that
\(
T_{qm}(\mathscr{C})\subseteq\mathscr{C}.
\)
Suppose first that \(H\in\mathscr{C}\). Since
\(
L^{(0)}=H,
\)
we have \(L^{(0)}\in\mathscr{C}\). If
\(L^{(k)}\in\mathscr{C}\) for some \(k\ge0\), then the invariance
assumption gives
\(
L^{(k+1)}
=
T_{qm}(L^{(k)})
\in
\mathscr{C}.
\)
Therefore, by induction,
\(
H\in\mathscr{C}
\Longrightarrow
L^{(n)}\in\mathscr{C}
\)
for every \(n\ge0.\)

The converse implication has already been proved in
\eqref{eq:backward-class-transfer-general}. Hence, for every finite
\(n\ge0\),
\(
H\in\mathscr{C}
\Longleftrightarrow
L^{(n)}\in\mathscr{C}.
\)
This proves~\eqref{eq:finite-time-class-equivalence}.
\end{proof}

\begin {remark}% NEW - simply text made as remark
Theorem~\ref{thm:finite-time-class-preservation} separates two
independent mechanisms. The recovery closure properties prevent a
finite Picard iterate from entering a class that does not contain
the initial language, while forward invariance prevents an orbit
starting inside the class from leaving it. When both mechanisms are
available, class membership becomes an exact finite-time invariant
of the dynamics.
\end{remark}

The abstract preservation principle established above can now be
applied to the successive classes of the Chomsky hierarchy. This
allows us to determine not merely whether a finite Picard iterate
belongs to a given language class, but whether it remains in the same
exact layer of the hierarchy as the initial language.

For the classification below, we use the Chomsky classes from
Definition~\ref{def:prelim-chomsky-hierarchy} 
for languages %NEW
over \(\Gamma\) and set
\[
\begin{aligned}
\operatorname{Type}(-1)&=\mathcal{L},
&
\operatorname{Type}(0)&=\mathcal{L}_{\mathrm{un}},\\
\operatorname{Type}(1)&=\mathcal{L}_{\mathrm{cs}},
&
\operatorname{Type}(2)&=\mathcal{L}_{\mathrm{cf}},\\
\operatorname{Type}(3)&=\mathcal{L}_{\mathrm{reg}},
&
\operatorname{Type}(4)&=\varnothing.
\end{aligned}
\]
The exact layers of the extended hierarchy are
\(
\mathcal{D}_i=
\operatorname{Type}(i)\setminus\operatorname{Type}(i+1),
\)
where \(i\in\{-1,0,1,2,3\}\).

\begin{theorem}\label{thm:exact-chomsky-layer-preservation}
For the marked Picard orbit 
of $T_{qm}$ % defined above,  %??????????????????????????????????? dali e dobre ?
the following assertions
hold.
\begin{enumerate}
\item[(i)]
If \(S\in\mathcal{L}_{\mathrm{cf}}\), then for every
\(i\in\{-1,0,1,2\}\) and \(n\ge0\),
\begin{equation}
L^{(0)} % H
\in\mathcal{D}_i
\quad\Longrightarrow\quad
L^{(n)}\in\mathcal{D}_i.
\label{eq:cfl-seed-layer-preservation}
\end{equation}

\item[(ii)]
If \(S\in\mathcal{L}_{\mathrm{reg}}\), then for every
\(i\in\{-1,0,1,2,3\}\) and \(n\ge0\),
\begin{equation}
L^{(0)} % H
\in\mathcal{D}_i
\quad\Longrightarrow\quad
L^{(n)}\in\mathcal{D}_i.
\label{eq:regular-seed-layer-preservation}
\end{equation}
\end{enumerate}
In both cases,
\(
L^{(n)}\to L_q^{(\star)}\in\mathcal{L}_{\mathrm{cf}},
\)
and the limit depends only on \(T_{qm}\), not on the initial
language \(
L^{(0)} % H
\).
\end{theorem}

\begin{proof}
We apply Theorem~\ref{thm:finite-time-class-preservation} to the
successive classes of the extended Chomsky hierarchy.

By Theorem~\ref{thm:prelim-chomsky-closure}, each of the classes
\(
\mathcal{L}_{\mathrm{reg}}, \ 
\mathcal{L}_{\mathrm{cf}}, \ 
\mathcal{L}_{\mathrm{cs}}, \ 
\mathcal{L}_{\mathrm{un}}
\)
satisfies the recovery assumptions~(R1)--(R2) of
Theorem~\ref{thm:finite-time-class-preservation}. The same is
trivially true for
\(
\operatorname{Type}(-1)=\mathcal{P}(\Gamma^\ast).
\)

Furthermore, Theorem~\ref{thm:prelim-chomsky-closure} implies that
whenever
\(
S\in\operatorname{Type}(i),
\)
the class \(\operatorname{Type}(i)\) is invariant under \(T_{qm}\).
Indeed, if \(L\in\operatorname{Type}(i)\), then
\(
(xLy)^q\in\operatorname{Type}(i),
\)
and hence
\(
T_{qm}(L)=S\cup(xLy)^q\in\operatorname{Type}(i).
\)

We first prove part~(i). Assume that
\(
S\in\mathcal{L}_{\mathrm{cf}}
=
\operatorname{Type}(2).
\)

Consider first \(i=-1\). If
\(
H=L^{(0)}\in\mathcal{D}_{-1},
\)
then
\(
H\notin\mathcal{L}_{\mathrm{un}}.
\)
Since \(\mathcal{L}_{\mathrm{un}}\) satisfies the recovery assumptions of
Theorem~\ref{thm:finite-time-class-preservation}, its contrapositive
part gives
\(
L^{(n)}\notin\mathcal{L}_{\mathrm{un}}
\)
for every finite \(n\ge0\). Since every
\(
L^{(n)}\in\operatorname{Type}(-1)=\mathcal{L},
\)
we obtain
\(
L^{(n)}\in\mathcal{D}_{-1}
\)
for every finite \(n\ge0\).

Now let
\(
i\in\{0,1,2\}
\)
and suppose that
\(
H=L^{(0)}\in\mathcal{D}_i.
\)
Then
\(
H\in\operatorname{Type}(i)
\)
and
\(
H\notin\operatorname{Type}(i+1).
\)
Since
\(
S\in\operatorname{Type}(2),
\)
the inclusions of the Chomsky hierarchy imply
\(
S\in\operatorname{Type}(i).
\)
Hence \(\operatorname{Type}(i)\) is invariant under \(T_{qm}\), and
Theorem~\ref{thm:finite-time-class-preservation} gives
\(
L^{(n)}\in\operatorname{Type}(i)
\)
for every finite \(n\ge0\).

On the other hand, \(\operatorname{Type}(i+1)\) satisfies the recovery
assumptions of the same theorem. Since
\(
H\notin\operatorname{Type}(i+1),
\)
the contrapositive implication gives
\(
L^{(n)}\notin\operatorname{Type}(i+1)
\)
for every finite \(n\ge0\). Therefore
\(
L^{(n)}
\in
\operatorname{Type}(i)
\setminus
\operatorname{Type}(i+1)
=
\mathcal{D}_i
\)
for every finite \(n\ge0\). This proves part~(i).

We next prove part~(ii). Assume that
\(
S\in\mathcal{L}_{\mathrm{reg}}
=
\operatorname{Type}(3).
\)
Then
\(
S\in\operatorname{Type}(i)
\)
for every
\(
i\in\{-1,0,1,2,3\}.
\)

For
\(
i\in\{-1,0,1,2\},
\)
the argument from part~(i) applies without change and gives
\(
H\in\mathcal{D}_i
\Longrightarrow
L^{(n)}\in\mathcal{D}_i
\)
for every finite \(n\ge0.\)

It remains to consider \(i=3\). Since
\(
\mathcal{D}_3
=
\operatorname{Type}(3)\setminus\operatorname{Type}(4)
=
\mathcal{L}_{\mathrm{reg}},
\)
let
\(
H=L^{(0)}\in\mathcal{L}_{\mathrm{reg}}.
\)
Because \(S\) is regular and the regular languages are closed under
concatenation and finite union, the class
\(\mathcal{L}_{\mathrm{reg}}\) is invariant under \(T_{qm}\). Therefore
\(
L^{(n)}\in\mathcal{L}_{\mathrm{reg}}
=
\mathcal{D}_3
\)
for every \(n\ge0.\)
This argument also covers the case \(H=\varnothing\).

Finally, in part~(i) the seed is context-free by assumption, while in
part~(ii) it is regular and hence context-free. Consequently,
Theorem~\ref{thm:cfl-fixed-point} gives
\(
L_q^{(\star)}\in\mathcal{L}_{\mathrm{cf}}.
\)
By the contraction theorem,
\(
L^{(n)}\longrightarrow L_q^{(\star)}
\)
for every initial language \(L^{(0)}\). Since the fixed point is
unique, it depends only on the operator \(T_{qm}\) and not on the
initial language.

Thus every exact Chomsky layer specified in parts~(i) and~(ii) is
preserved at every finite Picard stage, whereas all corresponding
orbits converge to the same context-free fixed point.
\end{proof}

\begin{corollary}
\label{cor:proper-cfl-seed-fixed-point}
Assume that
\(
S\in
\mathcal{L}_{\mathrm{cf}}
\setminus
\mathcal{L}_{\mathrm{reg}}.
\)
For the marked operator
\(
T_{qm}(L)=S\cup(xLy)^q,
\)
the unique fixed point satisfies
\(
L_q^{(\star)}
\in
\mathcal{L}_{\mathrm{cf}}
\setminus
\mathcal{L}_{\mathrm{reg}}.
\)
\end{corollary}

\begin{proof}
By Theorem~\ref{thm:cfl-fixed-point},
\(
L_q^{(\star)}\in\mathcal{L}_{\mathrm{cf}}.
\)
Since
\(
L_q^{(\star)}
=
S\cup(xL_q^{(\star)} y)^q,
\)
every word in the recursive part contains the
new %  fresh 
symbols \(x\) and
\(y\), whereas every word of \(S\) belongs to \(\Sigma^\ast\).
Consequently,
\(
L_q^{(\star)}\cap\Sigma^\ast=S.
\)
If \(L_q^{(\star)}\) were regular, then, since \(\Sigma^\ast\) is regular
and the regular languages are closed under intersection, the language
\(
S=L_q^{(\star)}\cap\Sigma^\ast
\)
would also be regular. This contradicts
\(
S\notin\mathcal{L}_{\mathrm{reg}}.
\)
Hence
\(
L_q^{(\star)}
\in
\mathcal{L}_{\mathrm{cf}}
\setminus
\mathcal{L}_{\mathrm{reg}}.
\)
\end{proof}

\begin{corollary}
\label{cor:proper-cfl-seed-regular-initial}
Assume that
\(
S\in
\mathcal{L}_{\mathrm{cf}}
\setminus
\mathcal{L}_{\mathrm{reg}}
\)
and that
\(
L^{(0)}=H\in\mathcal{L}_{\mathrm{reg}}.
\)
For the marked operator
\(
T_{qm}(L)=S\cup(xLy)^q,
\)
one has
\(
L^{(0)}\in\mathcal{L}_{\mathrm{reg}},
\)
while
\(
L^{(n)}
\in
\mathcal{L}_{\mathrm{cf}}
\setminus
\mathcal{L}_{\mathrm{reg}}\)
for every \(n\ge1.\)
\end{corollary}

\begin{proof}
Since both \(S\) and \(H\) are context-free, closure under
concatenation and finite union gives
\(L^{(n)}\in\mathcal{L}_{\mathrm{cf}}\)
for every \(n\ge0.\)

For every \(n\ge1\), the recursive part of
\(
L^{(n)}
=
S\cup(xL^{(n-1)}y)^q
\)
contains the fresh marker symbols \(x\) and \(y\), whereas every word
of \(S\) belongs to \(\Sigma^\ast\). Consequently,
\(L^{(n)}\cap\Sigma^\ast=S\)
for every \(n\ge1.\)

Suppose that \(L^{(n)}\) were regular for some \(n\ge1\). Since
\(\Sigma^\ast\) is regular and the regular languages are closed under
intersection, it would follow that
\(
S=L^{(n)}\cap\Sigma^\ast
\)
is regular, contradicting
\(
S\notin\mathcal{L}_{\mathrm{reg}}.
\)
Hence every \(L^{(n)}\), \(n\ge1\), is context-free but nonregular.
\end{proof}

The conclusions of Theorem~\ref{thm:exact-chomsky-layer-preservation}
and Corollaries~\ref{cor:proper-cfl-seed-fixed-point}
and~\ref{cor:proper-cfl-seed-regular-initial}
are summarized in Tables~\ref{tab:chomsky-orbit-regular-seed}
and~\ref{tab:chomsky-orbit-cfl-seed}.
In each table, the first column gives the class of the initial
language, the second describes the finite Picard iterates, and the
third gives the class of the common fixed point.
The first table covers the case of a regular seed, while the second
covers the case of a genuinely context-free seed.

\begin{table}[ht]
\centering
\caption{Case I: Classification of the finite Picard iterates and of the fixed point when \(S\in\mathcal L_{\mathrm{reg}}\).}
\label{tab:chomsky-orbit-regular-seed}

\small
\renewcommand{\arraystretch}{1.35}
\setlength{\tabcolsep}{4pt}

\begin{tabularx}{\textwidth}{
>{\raggedright\arraybackslash}X
>{\centering\arraybackslash}X
>{\centering\arraybackslash}X
}
\toprule
\textbf{Class of \(L^{(0)}\)}
&
\textbf{Class of \(L^{(n)}\)}
&
\textbf{Class of \(L_q^{(\star)}\)}
\\
\midrule

\(\mathcal L_{\mathrm{reg}}\)
&
\(\mathcal L_{\mathrm{reg}}\)
&
\(\mathcal L_{\mathrm{cf}}\)
\\

\(\mathcal L_{\mathrm{cf}}\setminus\mathcal L_{\mathrm{reg}}\)
&
\(\mathcal L_{\mathrm{cf}}\setminus\mathcal L_{\mathrm{reg}}\)
&
\(\mathcal L_{\mathrm{cf}}\)
\\

\(\mathcal L_{\mathrm{cs}}\setminus\mathcal L_{\mathrm{cf}}\)
&
\(\mathcal L_{\mathrm{cs}}\setminus\mathcal L_{\mathrm{cf}}\)
&
\(\mathcal L_{\mathrm{cf}}\)
\\

\(\mathcal L_{\mathrm{un}}\setminus\mathcal L_{\mathrm{cs}}\)
&
\(\mathcal L_{\mathrm{un}}\setminus\mathcal L_{\mathrm{cs}}\)
&
\(\mathcal L_{\mathrm{cf}}\)
\\

\(\mathcal L_{\mathrm{nrn}}=\mathcal L\setminus\mathcal L_{\mathrm{un}}\)
&
\(\mathcal L_{\mathrm{nrn}}\)
&
\(\mathcal L_{\mathrm{cf}}\)
\\

\bottomrule
\end{tabularx}
\end{table}

\begin{table}[ht]
\centering
\caption{Case II: Classification of the finite Picard iterates and of
the fixed point when
\(S\in\mathcal L_{\mathrm{cf}}\setminus\mathcal L_{\mathrm{reg}}\).}
\label{tab:chomsky-orbit-cfl-seed}

\small
\renewcommand{\arraystretch}{1.35}
\setlength{\tabcolsep}{4pt}

\begin{tabularx}{\textwidth}{
>{\raggedright\arraybackslash}X
>{\centering\arraybackslash}X
>{\centering\arraybackslash}X
}
\toprule
\textbf{Class of \(L^{(0)}\)}
&
\textbf{Class of \(L^{(n)}\)}\par
\(n\ge1\)
&
\textbf{Class of \(L_q^{(\star)}\)}
\\
\midrule

\(\mathcal L_{\mathrm{reg}}\)
&
\(\mathcal L_{\mathrm{cf}}\setminus\mathcal L_{\mathrm{reg}}\)
&
\(\mathcal L_{\mathrm{cf}}\setminus\mathcal L_{\mathrm{reg}}\)
\\

\(\mathcal L_{\mathrm{cf}}\setminus\mathcal L_{\mathrm{reg}}\)
&
\(\mathcal L_{\mathrm{cf}}\setminus\mathcal L_{\mathrm{reg}}\)
&
\(\mathcal L_{\mathrm{cf}}\setminus\mathcal L_{\mathrm{reg}}\)
\\

\(\mathcal L_{\mathrm{cs}}\setminus\mathcal L_{\mathrm{cf}}\)
&
\(\mathcal L_{\mathrm{cs}}\setminus\mathcal L_{\mathrm{cf}}\)
&
\(\mathcal L_{\mathrm{cf}}\setminus\mathcal L_{\mathrm{reg}}\)
\\

\(\mathcal L_{\mathrm{un}}\setminus\mathcal L_{\mathrm{cs}}\)
&
\(\mathcal L_{\mathrm{un}}\setminus\mathcal L_{\mathrm{cs}}\)
&
\(\mathcal L_{\mathrm{cf}}\setminus\mathcal L_{\mathrm{reg}}\)
\\

\(\mathcal L_{\mathrm{nrn}}=\mathcal L\setminus\mathcal L_{\mathrm{un}}\)
&
\(\mathcal L_{\mathrm{nrn}}\)
&
\(\mathcal L_{\mathrm{cf}}\setminus\mathcal L_{\mathrm{reg}}\)
\\

\bottomrule
\end{tabularx}
\end{table}

%===========================================================
% Examples
%===========================================================
\section{Examples}
\label{sec:examples}

In each of the following examples, we take \(q=2\) and denote the
corresponding marked operator \(T_{qm}\) from
Definition~\ref{def:marked-q-power} by \(T_{2m}\).

We conclude with four examples illustrating the main structural
phenomena established above. The first gives an explicit
context-free fixed point for the marked \(q\)-power dynamics.
The second illustrates the finite-time trace and exact recovery of
the initial language. The third shows preservation of every exact
layer of the extended Chomsky hierarchy when the seed is regular.
The fourth considers a genuinely context-free seed and a regular
initial language, showing that the orbit leaves the regular class
after the first Picard step and remains genuinely context-free at
every subsequent finite stage.

\begin{example}\label{ex:explicit-cfl-fixed-point}

Let \(\Sigma=\{c\}\), let \(x,y\notin\Sigma\) be two distinct new
symbols, and put \(\Gamma=\Sigma\cup\{x,y\}\). Take the regular seed
\(S=\{c\}\), let \(q=2\), and consider the operator
\(
T_{2m}:\mathcal{P}(\Gamma^\ast)\to\mathcal{P}(\Gamma^\ast)
\)
defined by
\(
T_{2m}(L)=\{c\}\cup(xLy)^2.
\)
The single guard is \((x,y)\), so its length is
\(
m=|x|+|y|=2.
\)
Hence Theorem~\ref{thm:prelim-q-power-contraction} gives
\(
d(T_{2m}(L),T_{2m}(M))\le2^{-4}d(L,M),
\)
and therefore \(T_{2m}\) has a unique fixed point \(L_2^{(\star)}\).

Since the seed is finite and hence context-free,
Theorem~\ref{thm:cfl-fixed-point} implies that
\(L_2^{(\star)}\in\mathcal{L}_{\mathrm{cf}}\). The fixed-point grammar can
be written particularly simply. Introduce two nonterminals \(X\) and
\(B\), take \(X\) as the start symbol, and use the productions
\[
X\to c\mid BB,
\qquad
B\to xXy.
\]
If \(K=L(X)\), then \(B\) generates \(xKy\), while \(X\to BB\)
generates \((xKy)^2\). Consequently,
\(
K=\{c\}\cup(xKy)^2=T_{2m}(K).
\)
By uniqueness of the fixed point, \(K=L_2^{(\star)}\).

For instance, starting from \(L^{(0)}=\varnothing\) gives
\(
L^{(1)}=\{c\}
\)
and
\(
L^{(2)}=\{c,xcyxcy\}.
\)
The finite iterates depend on the chosen initial language, whereas
the fixed point does not. Every Picard orbit of this operator
converges to the same context-free language \(L_2^{(\star)}\).

This example illustrates the first main phenomenon of the paper:
the language class of the limiting fixed point is determined by the
operator and its seed, independently of the language from which the
Picard iteration starts.
\end{example}

\begin{example}
\label{ex:finite-time-recovery}

We next illustrate Theorems~\ref{thm:finite-time-trace}
and~\ref{thm:finite-time-recoverability}. Let
\(
\Sigma=\{a,s\},
\)
let \(x,y\notin\Sigma\) be two distinct new symbols, and put
\(
\Gamma=\Sigma\cup\{x,y\}.
\)
Take the regular seed \(S=\{s\}\), let \(q=2\), and define
\(
T_{2m}(L)=\{s\}\cup(xLy)^2.
\)
Let
\(
H=L^{(0)}=\{a,aa\}\subseteq\Sigma^\ast
\)
and choose the distinguished word \(z=a\in H\).

At time \(n=0\), one has
\(
z_0=a,
\)
\(
\alpha_0=\beta_0=\varepsilon,
\)
and
\(
R_0=\Sigma^\ast.
\)
Hence
\(
L^{(0)}\cap R_0=H.
\)

At the first Picard step,
\(
xHy=\{xay,xaay\},
\)
and therefore
\[
L^{(1)}
=
\{s,\,
xayxay,\,
xayxaay,\,
xaayxay,\,
xaayxaay\}.
\]
The distinguished word evolves according to
\(
z_1=(xay)^2=xayxay.
\)
Moreover,
\(
\alpha_1=x
\)
and
\(
\beta_1=y(xay)=yxay.
\)
Thus
\(
R_1=x\Sigma^\ast yxay.
\)

Among the words of \(L^{(1)}\), exactly
\(xayxay\) and \(xaayxay\) belong to \(R_1\). Hence
\[
L^{(1)}\cap R_1
=
\{xayxay,xaayxay\}
=
x\{a,aa\}yxay
=
\alpha_1H\beta_1.
\]
The initial language is therefore recovered exactly by
\[
x^{-1}
\bigl(L^{(1)}\cap R_1\bigr)
(yxay)^{-1}
=
\{a,aa\}
=
H.
\]

At the second step,
\(
\alpha_2=xx
\)
and
\(
\beta_2=\beta_1y(xz_1y).
\)
Hence
\(
R_2=xx\Sigma^\ast\beta_2.
\)
The trace identity gives directly
\(
L^{(2)}\cap R_2
=
xxH\beta_2
=
\{xxa\beta_2,xxaa\beta_2\}.
\)
Applying the left quotient by \(xx\) and the right quotient by
\(\beta_2\) again recovers \(H\).

The same mechanism holds at every finite stage:
\(
L^{(n)}\cap R_n
=
\alpha_nH\beta_n,
\)
and consequently
\(
H=
\left(
\alpha_n^{-1}
\bigl(L^{(n)}\cap R_n\bigr)
\right)
\beta_n^{-1}.
\)
Thus the initial language remains exactly encoded in every finite
Picard iterate, even though the orbit converges to a fixed point that
is completely independent of \(H\).

The example also makes the injectivity result transparent. If
\(H_1,H_2\subseteq\Sigma^\ast\) contain the same distinguished word
\(z\) and satisfy
\(
T_{2m}^n(H_1)=T_{2m}^n(H_2)
\)
for some finite \(n\), then the same regular slice and the same
fixed-word quotients recover both initial languages. Hence
\(H_1=H_2\).
\end{example}

\begin{example}\label{ex:exact-chomsky-layer-preservation}
We now illustrate Theorem~\ref{thm:exact-chomsky-layer-preservation} simultaneously across the extended Chomsky hierarchy.

Let
\(
\Sigma=\{a,b,c,s\},
\)
let \(x,y\notin\Sigma\) be distinct new symbols, and put
\(
\Gamma=\Sigma\cup\{x,y\}.
\)
Take the regular seed \(S=\{s\}\), let \(q=2\), and consider
\(
T_{2m}(L)=\{s\}\cup(xLy)^2.
\)
Since the seed is regular, part~(ii) of
Theorem~\ref{thm:exact-chomsky-layer-preservation} applies to every
layer of the extended hierarchy.

Consider first the regular initial language
\(
H_3=a^\ast.
\)
If \(L_{H_3}^{(0)}=H_3\), then
\(
L_{H_3}^{(n)}\in\mathcal{L}_{\mathrm{reg}}\)
for every finite \(n\ge0.\)

Next, let
\(
H_2=\{a^k b^k:k\ge1\}.
\)
This language is context-free but not regular. Therefore
\(
L_{H_2}^{(n)}
\in
\mathcal{L}_{\mathrm{cf}}
\setminus
\mathcal{L}_{\mathrm{reg}}\)
for every finite \(n\ge0,\)
where \(L_{H_2}^{(0)}=H_2\).

Now take
\(
H_1=\{a^k b^k c^k:k\ge1\}.
\)
The language \(H_1\) is context-sensitive but not context-free.
Hence, if \(L_{H_1}^{(0)}=H_1\), then
\(
L_{H_1}^{(n)}
\in
\mathcal{L}_{\mathrm{cs}}
\setminus
\mathcal{L}_{\mathrm{cf}}\)
for every finite \(n\ge0.\)

By the strictness of the Chomsky hierarchy, choose a nonempty language
\(
H_0
\in
\mathcal{L}_{\mathrm{un}}
\setminus
\mathcal{L}_{\mathrm{cs}}.
\)
Then the corresponding orbit satisfies
\(
L_{H_0}^{(n)}
\in
\mathcal{L}_{\mathrm{un}}
\setminus
\mathcal{L}_{\mathrm{cs}}\)
for every finite \(n\ge0.\)

Finally, choose a language
\(
H_{-1}
\in
\mathcal{P}(\Sigma^\ast)
\setminus
\mathcal{L}_{\mathrm{un}}.
\)
For the Picard orbit starting from \(H_{-1}\),
Theorem~\ref{thm:exact-chomsky-layer-preservation} yields
\(
L_{H_{-1}}^{(n)}
\in
\mathcal{P}(\Gamma^\ast)
\setminus
\mathcal{L}_{\mathrm{un}}\)
for every finite \(n\ge0.\)

Thus, with a regular seed, every exact layer of the extended Chomsky
hierarchy is preserved at every finite Picard step. In particular,
the finite dynamics cannot move an initial language into a strictly
smaller Chomsky class.

The limiting behaviour is fundamentally different. Since the seed
\(S=\{s\}\) is context-free,
Theorem~\ref{thm:cfl-fixed-point} gives
\(
L_2^{(\star)}\in\mathcal{L}_{\mathrm{cf}}.
\)
All five Picard orbits above converge to this same fixed point.
Consequently, for example,
\(
L_{H_1}^{(n)}
\in
\mathcal{L}_{\mathrm{cs}}
\setminus
\mathcal{L}_{\mathrm{cf}}\)
for every finite \(n,\)
\(
L_{H_1}^{(n)}\longrightarrow
L_2^{(\star)}\in\mathcal{L}_{\mathrm{cf}}.
\)
Likewise, an orbit may remain outside the recursively enumerable
class at every finite stage and nevertheless converge to the same
context-free fixed point.

This is the finite-time versus limiting phenomenon captured by the
central theorem: the exact Chomsky layer of the initial language is
preserved throughout every finite stage, while the limiting language
may belong to a strictly smaller class.
\end{example}

\begin{example}\label{ex:cfl-seed-regular-initial}

The preceding example uses a regular seed and shows preservation of
every exact layer of the extended Chomsky hierarchy. We now illustrate
the different behaviour produced by a genuinely context-free seed when
the initial language is regular.

Let
\(
\Sigma=\{a,b\},
\)
let \(x,y\notin\Sigma\) be two distinct fresh symbols, and put
\(
\Gamma=\Sigma\cup\{x,y\}.
\)
Take
\(
S=\{a^k b^k:k\ge1\}
\in
\mathcal{L}_{\mathrm{cf}}
\setminus
\mathcal{L}_{\mathrm{reg}},
\)
let \(q=2\), and consider
\(
T_{2m}(L)=S\cup(xLy)^2.
\)
Choose the regular initial language
\(
L^{(0)}=\{\varepsilon\}.
\)

At the first Picard step,
\(
L^{(1)}
=
S\cup\{xyxy\},
\)
and therefore
\(
L^{(1)}
\in
\mathcal{L}_{\mathrm{cf}}
\setminus
\mathcal{L}_{\mathrm{reg}}.
\)

In fact, the same conclusion holds at every later finite step.
For each \(n\ge1\),
\(
L^{(n)}
=
S\cup(xL^{(n-1)}y)^2.
\)
Every word in the recursive part contains the marker symbols \(x\)
and \(y\), whereas every word of \(S\) belongs to \(\Sigma^\ast\).
Hence
\(
L^{(n)}\cap\Sigma^\ast=S\)
for every \(n\ge1.\)
If some \(L^{(n)}\), \(n\ge1\), were regular, then its intersection
with the regular language \(\Sigma^\ast\) would also be regular.
This would imply that \(S\) is regular, a contradiction. Consequently,
\(
L^{(n)}
\in
\mathcal{L}_{\mathrm{cf}}
\setminus
\mathcal{L}_{\mathrm{reg}}\)
for every \(n\ge1.\)

Thus a regular initial language leaves the regular class immediately
when the seed is genuinely context-free, and all subsequent finite
Picard iterates remain genuinely context-free. Moreover,
Corollary~\ref{cor:proper-cfl-seed-fixed-point} gives
\(
L_2^{(\star)}
\in
\mathcal{L}_{\mathrm{cf}}
\setminus
\mathcal{L}_{\mathrm{reg}},
\)
and the orbit converges to this same genuinely context-free fixed
point.

This example illustrates
Corollaries~\ref{cor:proper-cfl-seed-fixed-point}
and~\ref{cor:proper-cfl-seed-regular-initial} and completes the
classification summarized in
Tables~\ref{tab:chomsky-orbit-regular-seed}
and~\ref{tab:chomsky-orbit-cfl-seed}.
\end{example}

%===========================================================
% References
%===========================================================

\vskip 12pt

\begin{footnotesize}

\end{footnotesize}

\end{document}